\documentclass[a4paper,9pt]{scrartcl}

\usepackage{lineno,hyperref}
\modulolinenumbers[5]

\usepackage{amsmath,amssymb}

\usepackage{changepage}

\usepackage[utf8x]{inputenc}

\usepackage{textcomp,marvosym}

\usepackage{nameref,hyperref}

\usepackage[table]{xcolor}

\usepackage{array}

\usepackage{lastpage,fancyhdr,graphicx}
\usepackage{epstopdf}

\usepackage{longtable}
\usepackage{booktabs}
\usepackage{soul, color} 
\usepackage{xcolor}
\usepackage{multirow}
\usepackage{placeins}
\usepackage{float}
\usepackage{amssymb}         %additional math symbols (see amsguide.tex)
\usepackage{amsthm}
\graphicspath{{./images/}}
\usepackage{comment}
\usepackage{mathtools}
\usepackage[normalem]{ulem}
\usepackage[table]{xcolor}

\usepackage[a4paper]{geometry} % pour récupérer les marges
\usepackage{ragged2e} % pour \justifying
\justifying

\usepackage[vlined]{algorithm2e}
\SetKwInOut{Input}{Data}
\SetVlineSkip{1pt}

\definecolor{lightblue}{rgb}{0.8,0.9,1} % bleu ciel
\definecolor{lightred}{rgb}{1,0.5,0.4} % bleu ciel

\newcommand{\hal}[2]{#1} %\hal{à mettre dans hal}{à mettre dans autre que hal}

\newtheorem{Def}{Definition}
\newtheorem{Thm}{Theorem}

\newtheorem{Lemme}{Lemme}

\newtheorem{remark}{Remark}
\newtheorem{Proof}{Proof}

\begin{document}

%\begin{frontmatter}

\title{Lower Bound for (Sum) Coloring Problem}

\author{Alexandre Gondran\\ENAC, French Civil Aviation University, Toulouse, France\\\texttt{alexandre.gondran@enac.fr}%\inst{1}%\orcidID{0003-0116-021X} 
\and
Vincent Duchamp\\ENAC, French Civil Aviation University, Toulouse, France\\\texttt{vincent.duchamp@alumni.enac.fr}%\inst{1}%\orcidID{0003-0116-021X} 
\and
Laurent Moalic\\UHA,  University of Upper Alsace, Mulhouse, France\\\texttt{laurent.moalic@uha.fr}%\inst{2}%\orcidID{0003-3749-3227}
}
%
%\authorrunning{A. Gondran et al.}

%\institute{ENAC, French Civil Aviation University, Toulouse, France
%\email{alexandre.gondran@enac.fr}\and
%UHA,  University of Upper Alsace, Mulhouse, France\\
%\email{laurent.moalic@uha.fr}}
%
\date{}
\maketitle              % typeset the header of the contribution



\begin{abstract}
The Minimum Sum Coloring Problem is a variant of the Graph Vertex Coloring Problem, for which each color has a weight. 
This paper presents a new way to find a lower bound of this problem, based on a relaxation into an integer partition problem with additional constraints. 
We improve the lower bound for 18 graphs of standard benchmark DIMACS, and prove the optimal value for 4 graphs by reaching their known upper bound.

\end{abstract}

%\begin{keyword}
%Graph Coloring; Relaxation; Lower Bound; Integer Partition Problem
%%\texttt{elsarticle.cls}\sep \LaTeX\sep Elsevier \sep template
%%\MSC[2010] 00-01\sep  99-00
%\end{keyword}

%\end{frontmatter}

%\linenumbers

% Please keep the Author Summary between 150 and 200 words
% Use first person. PLOS ONE authors please skip this step. 
% Author Summary not valid for PLOS ONE submissions.   
%\section*{Author summary}
%Lorem ipsum dolor sit amet, consectetur adipiscing elit. Curabitur eget porta erat. Morbi consectetur est vel gravida pretium. Suspendisse ut dui eu ante cursus gravida non sed sem. Nullam sapien tellus, commodo id velit id, eleifend volutpat quam. Phasellus mauris velit, dapibus finibus elementum vel, pulvinar non tellus. Nunc pellentesque pretium diam, quis maximus dolor faucibus id. Nunc convallis sodales ante, ut ullamcorper est egestas vitae. Nam sit amet enim ultrices, ultrices elit pulvinar, volutpat risus.

%\linenumbers

% Use "Eq" instead of "Equation" for equation citations.
\section{Introduction}\label{intro}

The Minimum Sum Coloring Problem $MSCP$ is a variant of the Graph Vertex Coloring Problem (GVCP), with weights associated to colors. 
This problem can be applied to various domains such as scheduling, resource allocation or VLSI design~\cite{Malafiejski2004,Kubale2004}.

Given an undirected graph $G=(V,E)$ with $V$ a set of $n$ vertices and $E\subset V^2$ a set of edges, graph vertex coloring involves assigning each vertex with a color so that two adjacent vertices (linked by an edge) feature different colors. 
An equivalent formulation is to consider a coloring as a partition of $G$ into subsets of vertices so that two adjacent vertices not belong to the same subset\cite{Wu2018}. 

The GVCP consists in finding the minimum number of colors (or equivalently the minimum number of subsets), called \emph{chromatic number $\chi(G)$}, 
required to color (or equivalently to partition) the graph $G$.

%An equivalent view of a coloring with $k$ colors is to consider a partition of $G$ into $k$ subsets of vertices: $\{V_1,...,V_k\}$ so that two adjacent vertices not belong to the same subset. The number of colors used $k$ is a variable so that~: $1\leq k\leq n$.

%Usually each color is represented by an integer.
The MSCP is a variant of GVCP, in which each color has a cost equals to the integer that represents the color.
The objective of MSCP is to minimize the sum of the cost of the coloring, called \emph{chromatic sum} of $G$ and denoted $\Sigma(G)$. 
Figure~\ref{fig:MSCP} gives an example of MSCP from Jin and Hao~\cite{Jin2016} on a graph with $n=9$ vertices and shows the difference between the two problems.
\begin{figure}[H]
\centering
\includegraphics[width=0.7\linewidth]{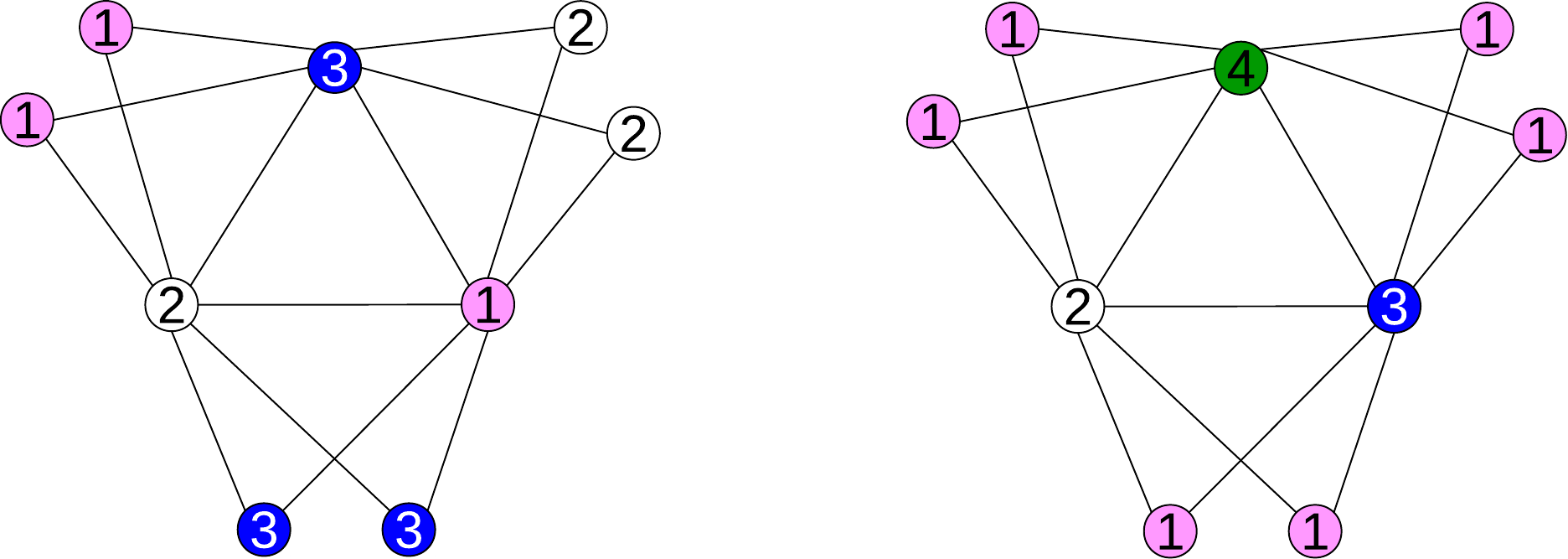}%{graphsPic.png}
\caption{Example of graph for MSCP from Jin and Hao~\cite{Jin2016}. The left coloring uses 3 colors (integers 1, 2 and 3).
It is an optimal solution of GVCP and the chromatic number $\chi(G)=3$. Moreover, its sum coloring cost is equal to 18. 
The right coloring uses one more color (integers 1, 2, 3 and 4) but its sum coloring cost is equal to 15, which is the chromatic sum of $G$: $\Sigma(G)=15$.}
\label{fig:MSCP}
\end{figure}

%Given $k$ a positive integer corresponding to the number of colors, a $k$-coloring of a given graph $G$ is a function $c$ that assigns to each vertex an integer between $1$ and $k$ as follows~:
%$$
%\begin{array}{rcl}
%c:V&\to&\{1,2...,k\}\\
%v&\mapsto&c(v)
%\end{array}
%$$
%The value $c(v)$ is called the color of vertex $v$. 
%The vertices assigned to the same color $i\in\{1,2...,k\}$ define a \emph{color class}, denoted $V_i$.
%An equivalent view is to consider a $k$-coloring as a partition of $G$ into $k$ subsets of vertices: 
%$c\equiv\{V_1,...,V_k\}$.
%Let a graph $G(V,E)$ with $V$ the set of $n$ vertices and $E \subset V\times V$ the edge set, color a graph we assign to each vertex $v \in V$ a color $c(v) \in \lbrace 1, ..., n \rbrace$.
%The objective is to find a coloring $c$ that is proper (if $(v_i, v_j) \in E$ then $c(v_i) \neq c(v_j)$) with the minimal sum of colors $f(c)$.
%\begin{equation}
%f(c)=\sum \limits_{{i=1}}^n c(v_i)= \sum \limits_{l=1}^n l|V_l|
%\end{equation}
%$|V_l|$ is the size (cardinality) of the independent set (or color class) of vertices using the color $l$ with $|V_1|\geq ...\geq |V_n|$.
%The minimal sum of color is called the chromatic sum of $G$ and is denoted, $\Sigma(G)$.
%\begin{equation}
%\Sigma (G) = \min \limits_{\lbrace c \mid c\textnormal{ is proper}\rbrace} f(c)
%\end{equation}
More precisely, one possible formulation of the MSCP is the following~:
%\begin{equation}
%(P)\left\{
%\begin{array}{rll}
%\displaystyle \text{Min.}   &   \displaystyle f(V_1, V_2..., V_n)= \sum \limits_{l=1}^n l|V_l|\\[1mm]
%\text{s.c.} &\displaystyle \cup_{l=1..n} V_l = V\label{eq:c1}\\[1mm]
%            &\displaystyle V_l \cap V_k = \emptyset,&\forall (l,k), l\neq k\\[1mm]
%            &\displaystyle |V_l| \geq |V_{l+1}|, &\forall l=1..n-1\\[1mm]
%            &\displaystyle \{i,j\} \nsubseteq V_l,&\forall (i,j) \in E,\ \forall l=1..n\\[1mm]
%            & V_l \subset V&\forall l=1..n
%\text{et}           & x \in \N^n \text{ ou } x \in \{0,1\}^n \\[1mm]
%\end{array}
%\right. 
%\end{equation}

%\begin{equation}
$$(MSCP)\left\{
\begin{minipage}[c]{0.8\linewidth}\vspace{-5mm}
\begin{eqnarray}%{rll}
\displaystyle \text{Min.}   &   \displaystyle f_\Sigma(\mathbf{V})= \sum \limits_{l=1}^n l|V_l|\label{eq:obj}\\[-0.5mm]
\text{s.c.} &\displaystyle \cup_{l=1..n} V_l = V\label{eq:partition1}\\[-0.5mm]
            &\displaystyle V_l \cap V_k = \emptyset,&\forall (l,k), l\neq k\label{eq:partition2}\\[-0.5mm]
            &\displaystyle |V_l| \geq |V_{l+1}|, &\forall l=1..n-1\label{eq:decrease}\\[-0.5mm]
            &\displaystyle \{i,j\} \nsubseteq V_l,&\forall (i,j) \in E,\ \forall l=1..n\label{eq:consColoring}\\[-0.5mm]
            & V_l \subset V&\forall l=1..n\label{eq:domain}%
%\text{et}           & x \in \N^n \text{ ou } x \in \{0,1\}^n \\[1mm]
\end{eqnarray}
\end{minipage}
\right. \nonumber
$$%\end{equation}

We denoted $\mathbf{V}=(V_1, V_2..., V_n)$ the partition of $G$ (constraints (2) and (3)) into $n$ sets with $n$, the number of vertices of $G$. %~\ref{eq:partition1} and~\ref{eq:partition2}). 
For all $1\leq i\leq n$, each $V_i$ is called color class.
To be as general as possible, we do not precise the number $k$ of colors used, we only know that $1\leq k\leq n$.
Then the number of colors used in a coloring is equal to the number of none empty color classes~: $k=|\{V_i\in \mathbf{V}\ |\ |V_i|\geq1\}|$.
Constraint (5) indicates that two adjacent vertices cannot be in the same color class. That is the coloring constraint.
Constraint (4) forces the color classes to be ordered from the largest to the smallest size.
With this convention, the objective function can be expressed as equation~(1).
An optimal sum coloring of graph $G$, noted $\mathbf{V}^*$, has its objective function equals to the chromatic sum: $f_\Sigma(\mathbf{V}^*)=\Sigma(G)$.

$MSCP$ is an NP-hard problem\cite{Kubicka1991} and exact methods to solve it are effective only on small instances or specific graphs. 
For general graphs, we use heuristics to get a sub-optimal solution\cite{Jin2017}. 
It provides an upper bound of the optimal solution. We can compare it to a lower bound to estimate the quality of the solution. 

This paper presents a new way to find lower bounds for this problem. 
We obtain it by relaxing MSCP into an Integer Partition Problem (IPP). 
A similar approach has recently be use by Lecat, Lucet and Li~\cite{Lecat2017_2}. 
They find a lower bound, called $LBM\Sigma$, using the notion of \emph{motif} that improve largely the best lower bound of the literature (DIMACS benchmarks~\cite{dimacs96}). 
This paper improves this lower bound by counting the maximal number of independent sets of maximal size in a graph.
An Independent Set (IS) or stable set is a set of vertices of $G$, no two of which are adjacent.
A color class is by definition an IS.
Others approaches designed to find the lower bound for $MSCP$ are proposed in~\cite{Wu2018,Jin2016,Wu2013}.

Experiments on standard benchmarks DIMACS~\cite{dimacs96} of graph instances show that we improve the lower bound for several graphs and we sometimes prove their optimal value by reaching their known upper bound.

In the following sections, we first present how we relax $MSCP$ to an $IPP$ (Section~\ref{sect:relaxation}). 
In Section~\ref{sect:solvingIPP}, we detail a new way of ordering integer partitions, and we use it to solve exactly our relaxed problem. 
Section~\ref{sect:LB_GVCP} shows how to extend this lower bound to GVCP.
We show and analyze our results in Section~\ref{sect:result}, and then we conclude. 
%In the appendix, we present the additional work that made during the PIR on the upper bounds of MSCP.

% For figure citations, please use "Fig" instead of "Figure".

% Place figure captions after the first paragraph in which they are cited.
%\begin{figure}[!h]
%\caption{{\bf Bold the figure title.}
%Figure caption text here, please use this space for the figure panel descriptions instead of using subfigure commands. A: Lorem ipsum dolor sit amet. B: Consectetur adipiscing elit.}
%\label{fig1}
%\end{figure}

% Results and Discussion can be combined.
%\section{Minimum Sum Coloring Problem}\label{section mscp}

\section{Relaxation as an integer partition problem}\label{sect:relaxation}

\subsection{Integer partitions}

We relax the problem of sum coloring $MSCP$ into a problem of integer partition, denoted $IPP_{\Sigma M}$, with the same objective function but with less constraints.

An integer partition is a way of writing a positive integer $n$ as a sum of other positive integers.
We say that the vector $\mathbf{a}=(a_i)_{1\leq i\leq n} \in \mathbb{N}^n$ is a partition of $n$ if~:
\begin{equation*}
(IP)\left\{
\begin{array}{l}
\displaystyle \sum_{i=1}^n a_i = n\\[1mm]
\displaystyle a_i \geq a_{i+1},\ \forall i=1..n-1\\[1mm] %\mathbf{a}=(a_i)_{1\leq i\leq n} \in \mathbb{N}^n
%\text{et}           & x \in \N^n \text{ ou } x \in \{0,1\}^n \\[1mm]
\end{array}
\right. 
\end{equation*}

We arbitrarily choose to rank $a_i$ in decreasing order. 
Partitions can be graphically visualized with Young diagrams~\cite{Young1900} or Ferrer diagrams. 
Figure~\ref{fig:youngDiagram_12} represents the partition of the integer $12=4+4+3+1+0+0+0+0+0+0+0+0$.
The vector of this partition is: $\mathbf{a}=(4,4,3,1,0,0,0,0,0,0,0,0)$ or simply noted $\mathbf{a}=(4,4,3,1)$ because the two first lines have four squares, the third line, three squares and the last one, one square.

\begin{figure}[H]
\centering
\includegraphics[width=0.1\linewidth]{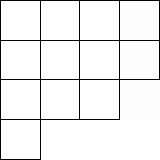}
\caption{Example of Young diagram for the partition of integer $12$ into a sum of four integers~: $4$, $4$, $3$ and $1$ ($12=4+4+3+1$). Each line corresponds to one positive integer of the partition. Integers are ranked in descending order.}
\label{fig:youngDiagram_12}
\end{figure}

\subsection{Definition of $IPP_{\Sigma M}$}

Instead of considering a set partition (or equivalently vertices partition), we focus on the cardinality of each set and we study an integer partition of the number $n=|V|$, corresponding to the number of vertices. We note for each vertices set $V_i$, its cardinality $a_i=|V_i|$.
The objective function (1) of $MSCP$ becomes $\sum_{i=1}^n ia_i$.

By this way, the two first constraints of $MSCP$ (constraints (2) and (3)) imply $\sum_{i=1}^n a_i = n$.
Constraint (4) of $MSCP$ is still valid with the new notations: $a_i\geq a_{i+1}$.
We represent a coloring with Young diagram where each square represents a vertex. 
The number of color used is the number of lines. 
Squares in the same line are in the same color class. 
Figure~\ref{fig:Young-coloring} represents Young diagrams of the two colorings presented in Figure~\ref{fig:MSCP}. 

\begin{figure}[H]
\centering
\includegraphics[width=0.6\linewidth]{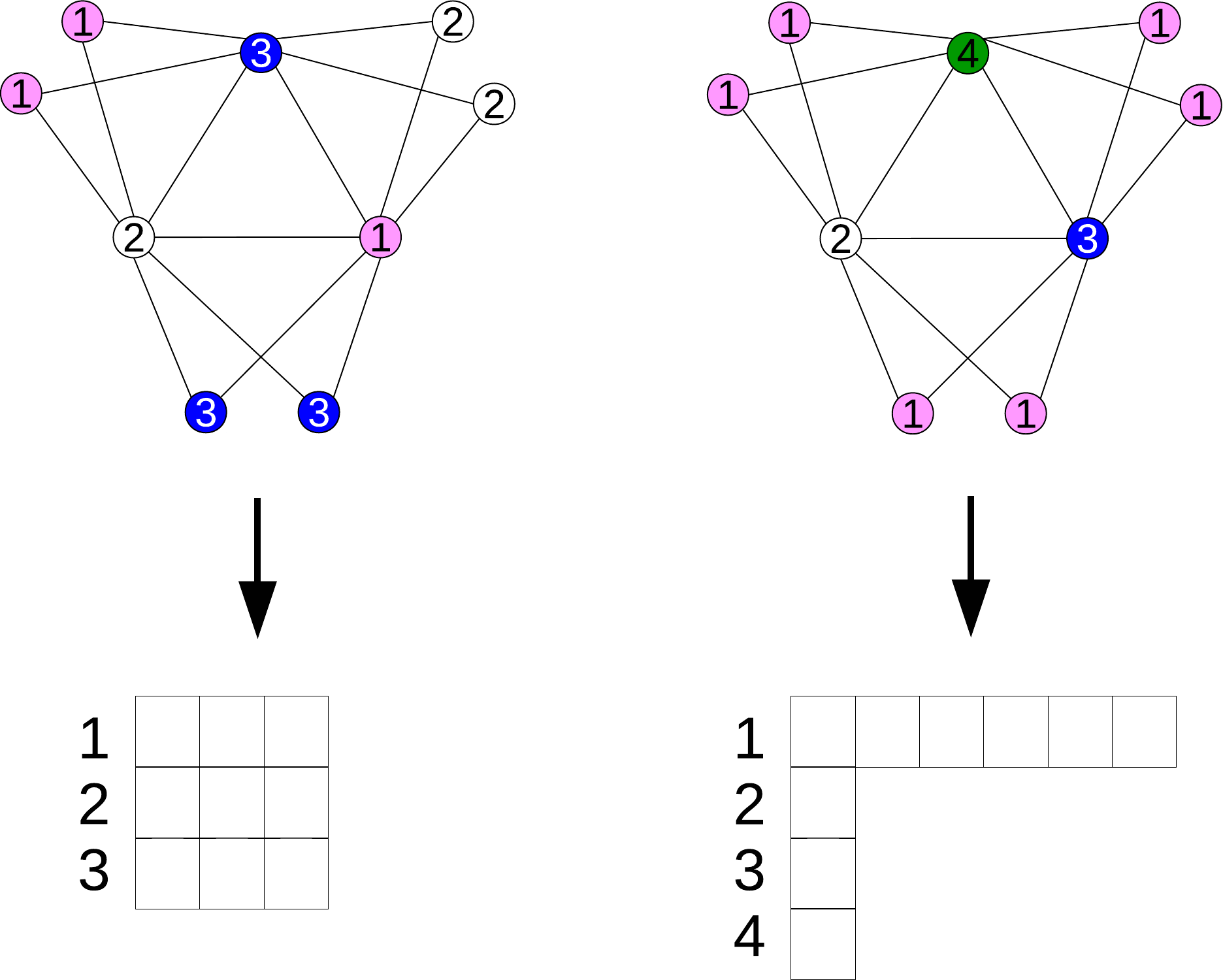}%{relaxation.png}
\caption{Young diagrams of two colorings of the same graph as in Figure~\ref{fig:MSCP}. 
Each square represents a vertex. Squares on the same line share the same color.
The objective function value of $MSCP$ is calculated by counting the number of squares for each line: $\sum_{i=1}^n ia_i$.}
\label{fig:Young-coloring}
\end{figure}

The \textit{stability number} of $G$, noted $\alpha(G)$, is the size of the maximum IS. 
It indicates the maximum number of columns of a Young diagram that it is possible to use for a coloring, therefore $a_i \leq \alpha(G)$.
Moreover, we will define $m$ as an upper bound of the maximum number of color classes of size $\alpha(G)$ which can be in a coloring.
Then, we can add an extra constraint~: $|\{a_i | a_i=\alpha(G)\}| \leq m$ or equivalently if $m<n$ to $a_{m+1}\leq\alpha(G)-1$.
Finding $\alpha(G)$ is a NP-hard problem~\cite{Garey79} and counting all maximal IS is \#-P complete~\cite{Valiant1979}. 
Therefore, there are no polynomial-time algorithm able to solve these problems, unless $P=NP$.
However, the graph size for which it is possible to solve them exactly in a reasonable time (less a few minutes) is around 1000 vertices for random graphs with $0.5$ density by running a solver such as 
%Cliquer~\cite{Ostergaard2001}, IncMaxCLQ~\cite{Li2013} or 
MoMC~\cite{Li2017}, the state-of-art exact maximum clique algorithm.
We define $\overline{\alpha}$ an upper bound of $\alpha(G)$, for which the above constraints are still true: $a_i \leq \overline{\alpha}$ and $|\{a_i | a_i=\overline{\alpha}\}| \leq m$.
If MoMC takes too much time to compute $\alpha(G)$, then a possible upper bound of the largest IS size, $\overline{\alpha}$, is the highest positive integer, $k$, verifying: $|\{v\in V,\, d_{\overline{G}}[v]\geq k\}|\geq k$ with $d_{\overline{G}}[v]$ the degree of the vertex $v$ in  $\overline{G}$, the complement graph of $G$\footnote{Notice that the maximum IS of a graph $G$ is the maximum clique of $\overline{G}$.}.
Another upper bound is the number of colors used for a vertex coloring of $\overline{G}$ not necessarily optimal.

The \textit{chromatic strength} of a graph $G$, noted $s(G)$, is the minimum number of colors used in all the optimal colorings of $MSCP$. 
It indicates the minimum number of lines of a Young diagram that are needed for a coloring. We note $\underline{s}$ a lower bound of $s(G)$ and we notice that~:
\begin{equation}
s(G)\geq\chi(G)\geq\left\lceil\frac{n}{\alpha(G)}\right\rceil
\end{equation}
Indeed the \textit{chromatic number} $\chi(G)$ is at least a lower bound of $s(G)$ because an optimal coloring of $MSCP$ is at least a legal coloring for GVCP.
Moreover, a simple lower bound of $\chi(G)$ is $\left\lceil\frac{n}{\alpha(G)}\right\rceil$ or at least $\left\lceil\frac{n}{\overline{\alpha}}\right\rceil$.
The number of coefficients of $\mathbf{a}$ not null (or equivalently the number of lines of Young diagram) is at least higher than $s(G)$ and therefore higher than $\underline{s}$, i.e. $|\{a_i | a_i\geq 1\}| \geq \underline{s}$ or equivalently $a_{\underline{s}}\geq 1$ because $(a_i)$ is a decreasing suite of integer.

%Indeed in the best case, all color classes of an optimal coloring of GVCP have the maximal size possible i.e. $\alpha(G)$.
%Then the number of color classes (i.e. number of color used) are at least upper than $\frac{n}{\alpha(G)}$.

%Therefore $MSCP$ becomes:

Given a graph $G=(V,E)$, knowing $n=|V|$ its size, $\overline{\alpha}$ an upper bound of its stability number, $\underline{s}$ a lower bound of its chromatic strength and $m$, an upper bound of the maximum number of color classes of size $\overline{\alpha}$, we define $IPP_{\Sigma M}$ problem as:
%\begin{equation*}
$$
(IPP_{\Sigma M})\left\{
\begin{minipage}[c]{0.7\linewidth}\vspace{-5mm}
\begin{eqnarray}
\displaystyle \text{Min.}   &   \displaystyle f_{\Sigma M}(\mathbf{a})= \sum \limits_{i=1}^n i a_i\\[-0.5mm]
\text{s.c.} &\displaystyle \sum_{i=1}^n a_i = n\label{eq:partition1bis}\\[-0.5mm]
            &\displaystyle a_i \geq a_{i+1}, &\forall i=1..n-1\label{eq:partition2bis}\\[-0.5mm]
            &\displaystyle a_i \leq \overline{\alpha},&\forall i=1..n\label{eq:coloringyy}\\[-0.5mm]
            &\displaystyle |\{a_i | a_i=\overline{\alpha}\}| \leq m\label{eq:constraintm}\\[-0.5mm]
            &a_{\underline{s}}\geq 1\label{eq:LBchi}\\[-0.5mm]
            & \mathbf{a}=(a_i)_{1\leq i\leq n} \in \mathbb{N}^n
%\text{et}           & x \in \N^n \text{ ou } x \in \{0,1\}^n \\[1mm]
\end{eqnarray}
\end{minipage}
\right. \nonumber
$$%\end{equation*}
$IPP_{\Sigma M}$ depends only of four integers: $n$, $\overline{\alpha}$, $\underline{s}$ and $m$.
We denote $\Sigma M$ the optimal objective function value of one of its optimal partition, $\mathbf{a}^*$: $f_{\Sigma M}(\mathbf{a}^*)=\Sigma M$.
Therefore, $\Sigma M$ is function of these four parameters: $\Sigma M(n, \overline{\alpha}, \underline{s}, m)$.
\begin{Thm}
  $IPP_{\Sigma M}$ is a relaxation of $MSCP$, that is, 
%the objective function value of optimal solutions of $IPP_{\Sigma M}$, denoted 
$\Sigma M(n, \overline{\alpha}, \underline{s}, m)$ is a lower bound of $\Sigma(G)$:
\begin{equation}
\Sigma M\leq \Sigma(G)
\end{equation}
\end{Thm}
\begin{Proof}
 The fourth constraint of $MSCP$ (coloring constraint (5)) imply constraints (11-13), then $IPP_{\Sigma M}$ is a relaxation of $MSCP$.

\end{Proof}

\subsection{Analysis}
%$(P_{\Sigma M})$ is a relaxation of $MSCP$, $(P)$, so the value of its optimal solution, noted $\Sigma M(n, m, \alpha(G), \underline{s}(G))$, is a lower bound of $\Sigma(G)$~:
%\begin{Thm}
%\end{Thm}

Note that if we remove the constraint $(13)$ from $IPP_{\Sigma M}$ (or equivalently if we fix $m$ at $+\infty$), we would get a problem equivalent to the one used by Lecat et al.~\cite{Lecat2017_2} to find the $LBM\Sigma$ lower bound: $LBM\Sigma=\Sigma M(n, \overline{\alpha}, \underline{s}, m=+\infty)$ with $\overline{\alpha}=\alpha(G)$, if it is known and~:
\begin{equation}
\underline{s}=
\left\{
\begin{array}{ll}
\chi(G)& \text{if it is known}\\
\left\lceil\frac{n}{\alpha(G)}\right\rceil&\text{otherwise}
\end{array}
\right.
\end{equation}
Therefore, $\Sigma M$ is an improvement of $LBM\Sigma$, obtained by using the integer $m$ corresponding to the maximum number of IS of size $\alpha(G)$ that is possible to use in a coloring of $G$ graph. 
%\begin{Thm}
\begin{equation}
LBM\Sigma \leq \Sigma M \leq \Sigma(G)
\end{equation}
%\end{Thm}
For this reason, we need to find the smallest possible value of $m$ to have the highest possible value of $\Sigma M$.
To obtain our value of $m$, we define a new graph called \emph{the maximum independent set graph}, noted $\tilde{G}=(\tilde{V},\,\tilde{E})$ as follow.
\begin{Def}
Let a graph $G=(V,E)$, we call the \textbf{maximum independent set graph} of $G$, the graph $\tilde{G}=(\tilde{V},\,\tilde{E})$ build as follow~: 
\begin{itemize}
  \item each vertex of $\tilde{V}$ is a maximum independent set\footnote{independent set of size $\alpha(G)$} of $G$;
  \item it exists an edge $e=(u,v)\in\tilde{E}$ between two independent sets $u\in\tilde{V}$ and $v\in\tilde{V}$, if and only if $u$ and $v$ have at least one vertex $w\in V$ of $G$ in common (i.e. $u\cap v\neq\emptyset$); we said that $u$ and $v$ are \emph{incompatible} because both can not be part of the same coloring of $G$.
\end{itemize}
\end{Def} 

\begin{Thm}
An optimal sum coloring $\mathbf{V^*}$, of a graph $G$ (i.e. an optimal solution of $MSCP$) can not have more than $\alpha(\tilde{G})$ color classes of size $\alpha(G)$~:
$$
|\{V_i\in\mathbf{V^*},\,|V_i|=\alpha(G)\}|\leq\alpha(\tilde{G})
$$
\end{Thm}
\begin{Proof}
Finding the maximum IS of $\tilde{G}$ corresponds to find the maximal number of maximal ISs compatibles in $G$.
In others words, $\alpha(\tilde{G})$ is the maximum number of IS of size $\alpha(G)$ that are possible to include in a same coloring of $G$. 
This is true not only for sum coloring problem but also for all coloring problems.
\end{Proof}

Therefore, we use in experimental tests, when it is possible to compute it:
\begin{equation}
m=\alpha(\tilde{G})
\end{equation}

\section{Solving $IPP$}\label{sect:solvingIPP}
\subsection{Ordering integer partitions}\label{sect:comparison}

We defined a relaxation problem, $IPP_{\Sigma M}$, of $MSCP$.
Our aim is to solve exactly this problem in order to provide a lower bound to $MSCP$.
We define an order between integer partitions corresponding to the objective function $f(\mathbf{a})= \sum \limits_{i=1}^n i a_i$.
Given a partition $\mathbf{a}$, we define a set of successor partitions if we move down only one square in the corresponding Young diagram of $\mathbf{a}$ by the following rule~:

\begin{algorithm}[H]
 \DontPrintSemicolon
% \SetVline
\Input{$n$}
% \\~\\
  \SetKwFunction{FMain}{$successor$}
  \SetKwProg{Fn}{Fonction}{:}{}
  \Fn{\FMain{$\mathbf{a}$}}{
  $succ \leftarrow \varnothing$\;
  \ForEach{$i=1..n-1$}{
  	\If{$a_i\neq1\ \wedge\ a_{i+1}<a_i$}{
    	$j\leftarrow i+1$\;
        \While{$a_j \geq a_i -1$}{
        	$j\leftarrow j+1$
        }
        $\mathbf{b}\leftarrow change(\mathbf{a},i,j)$\textit{ // we note~: $\mathbf{b}\leftarrow \mathbf{a} \oplus (i,j)$}\;
        $succ \leftarrow succ\cup\{\mathbf{b}\}$
    }
  }
  \Return $succ$
}
\label{algo:succ}\caption{Function that return all the successors of the partition $\mathbf{a}$.}
\end{algorithm}
The $change$ function, noticed $\oplus$ operator, is defined as follows~:

\begin{algorithm}[H]
 \DontPrintSemicolon
% \SetVline
\Input{$n$}
% \\~\\
  \SetKwFunction{FMain}{$change$}
  \SetKwProg{Fn}{Fonction}{:}{}
  \Fn{\FMain{$\mathbf{a}$,$i$,$j$}}{
  \ForEach{$k=1..n$}{
  		$b_k\leftarrow a_k$
	}
	$b_i\leftarrow a_i - 1$\;
    $b_j\leftarrow a_j + 1$\;
  \Return $\mathbf{b}$
}
\label{algo:change}\caption{Function that return a neighbor partition of $\mathbf{a}$, just two values of $\mathbf{b}$ differ from $\mathbf{a}$; we note $\mathbf{b}\leftarrow\mathbf{a}\oplus(i,j)$}
\end{algorithm}

Figure~\ref{fig:successor} illustrates the successor operator with Young diagram.
Algorithm~\ref{algo:succ} shows that it is possible to move down the last square of each line if~:
\begin{itemize}
 \item the following line has not the same number of squares; it is why the red square of first line of Figure~\ref{fig:successor} can not move.\\[-7mm]
 \item the line has not an unique square; it is why the orange square of last line of Figure~\ref{fig:successor} can not move. 
\end{itemize}
When it is possible to move down a square (case of blue and green squares of Figure~\ref{fig:successor}), 
the square takes last place of the first possible line.
Blue square of Figure~\ref{fig:successor} can take the place in the just following line. 
But green square of Figure~\ref{fig:successor} must go two lines down. 

\begin{remark}
By construction, if $\mathbf{a}$ is an admissible solution of $IPP_{\Sigma M}$, then all of $\mathbf{a}$'s succesor interger partitions are also an admissible solution of $IPP_{\Sigma M}$.
In other words, by noting $\Omega(IPP_{\Sigma M})$ the set of all admissible solutions of $IPP_{\Sigma M}$:
if $\mathbf{a}\in\Omega(IPP_{\Sigma M})$, then $succesor(\mathbf{a}) \subset \Omega(IPP_{\Sigma M})$.
\end{remark}

\begin{remark}
 If $\mathbf{a}$ is an admissible solution of $IPP_{\Sigma M}$ without succesor i.e. $successor(\mathbf{a})=\varnothing$, then it means that $\mathbf{a}$ is the column partition $\mathbf{a}=(\underbrace{1,..., 1}_n)$.
\end{remark}

\begin{figure}[H]
\centering
\includegraphics[width=0.5\linewidth]{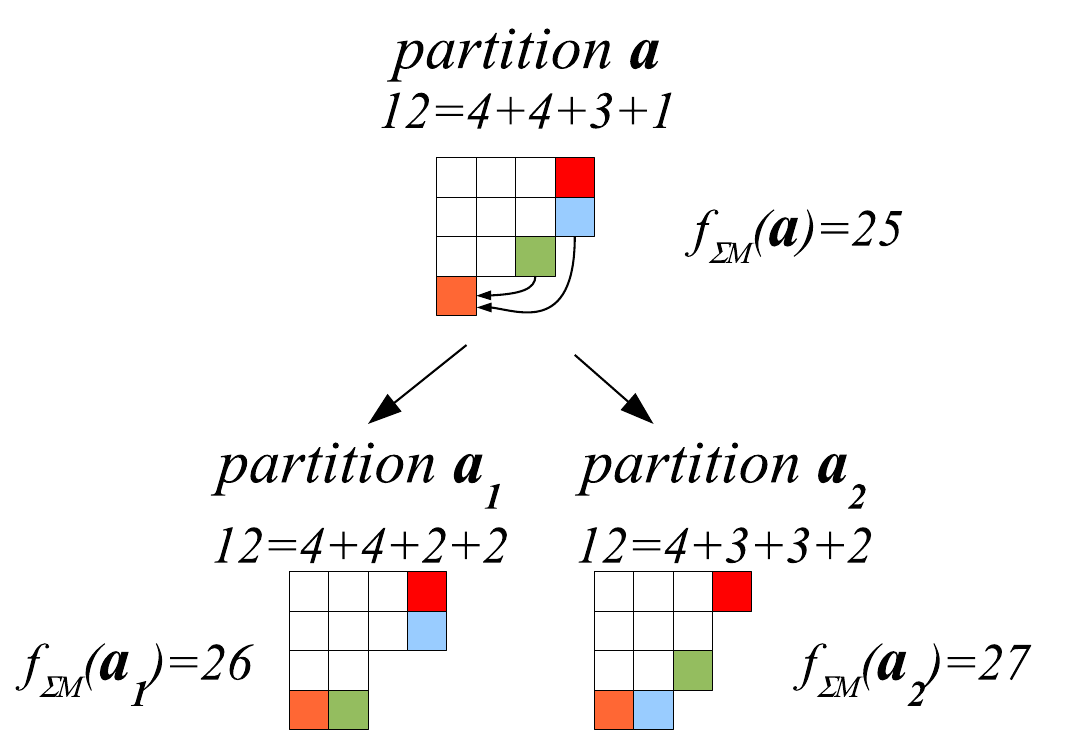}
\caption{Illustration for the way we find the two successors of a partition $\mathbf{a}$: $\mathbf{a_1}$ and $\mathbf{a_2}$. 
From $\mathbf{a}$, we move a square from a line $i(=2$ or $3)$ to the closest line $j$ where $a_i > a_j$; $j=4$.}
\label{fig:successor}
\end{figure}

We define symmetrically the $predecessor$ function~:

\begin{algorithm}[H]
 \DontPrintSemicolon
% \SetVline
\Input{$n$, $\overline{\alpha}$, $\underline{s}$, $m$}
% \\~\\
  \SetKwFunction{FMain}{$predecessor$}
  \SetKwProg{Fn}{Fonction}{:}{}
  \Fn{\FMain{$\mathbf{a}$}}{
  $pred \leftarrow \varnothing$\;
  \ForEach{$i=2..n-1$}{
  	\If{$ ~~~~~[(i\leq \underline{s}\ \wedge\ a_i>1)\ \ \vee\ \ (i> \underline{s}\ \wedge\ \ a_i>0)]\newline ~~~~\wedge\ \ [(i\leq m\ \wedge\ \ a_{i-1}<\overline{\alpha})\ \ \vee\ \ (i> m\ \wedge\ \ a_{i-1}<\overline{\alpha}-1)]\newline ~~~~\wedge\ \ a_{i+1}<a_i$}{
    	$j\leftarrow i-1$\;
        \While{$j\neq1\ \wedge\ a_j \geq a_{j -1}$}{
        	$j\leftarrow j-1$
        }
        $\mathbf{b}\leftarrow change(\mathbf{a},i,j)$\;
        $pred \leftarrow pred\cup\{\mathbf{b}\}$
    }
  }
  \Return $pred$
}
\label{algo:pred}\caption{Function that returns all the predecessors of the partition $\mathbf{a}$.}
\end{algorithm}

Figure~\ref{fig:predecessor} illustrates the predecessor operator with Young diagram.
Algorithm~\ref{algo:pred} shows that it is possible to move up the last square of each line $i>1$ if~:
\begin{itemize}
 \item the line just below $(i+1)$ has not the same number of squares; it is why the red squares of the two first lines of partition $\mathbf{a}$ of Figure~\ref{fig:predecessor}-up-right can not move.\\[-7mm]
 \item the line just above $(i-1)$ has strictly less than $\overline{\alpha}$ squares if $i-1\leq m$; it is why the red square of the third line of partition $\mathbf{a}$ of Figure~\ref{fig:predecessor}-up-right can not move.\\[-7mm]
 \item the line just above $(i-1)$ has strictly less than $\overline{\alpha}-1$ squares if $i-1>m$; this constraint and the previous one define the hashed area (forbidden area) of Figure~\ref{fig:predecessor}-up-right.\\[-7mm]
%it is why the green and blue squares of partition $\mathbf{b}$ of Figure~\ref{fig:predecessor} can move to line 4.\\[-7mm]
 \item the square of line $\underline{s}$ is alone on its line; it is why the red square of last line of partition $\mathbf{a}$ of Figure~\ref{fig:predecessor}-up-right can not move. This mandatory square is noted in bold on Figure~\ref{fig:predecessor}-up-right.
\end{itemize}
When it is possible to move a square up (case of blue and green squares of partition $\mathbf{b}$ of Figure~\ref{fig:predecessor}-center), 
the square takes last place of the first possible line (line 3 and 5 respectively).

\begin{figure}[H]
\centering
\includegraphics[width=0.7\linewidth]{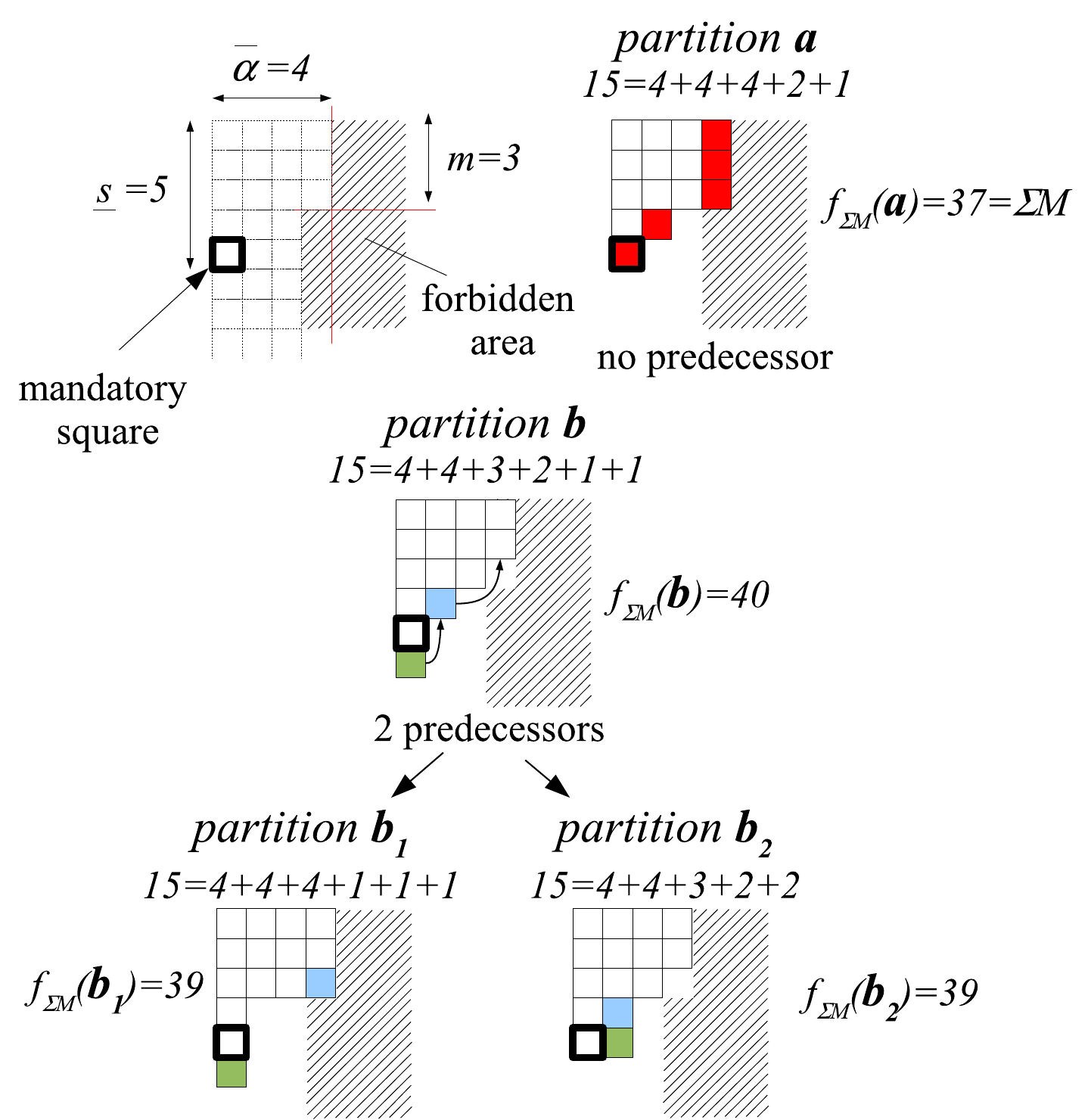}
\caption{Illustration of the way we find a predecessor of a partition.
The partition $\mathbf{a}$ (up-right figure) has no predecessor because the hashed area is the forbidden area and the bold square (last line) is mandatory (a square must be place on it).
The partition $\mathbf{b}$ has two possible predecessors, $\mathbf{b_1}$ and $\mathbf{b_2}$, by moving blue square or respectively green square to line 3 or respectively line 5.}
\label{fig:predecessor}
\end{figure}

\begin{remark}
 By construction, if $\mathbf{b}\in\Omega(IPP_{\Sigma M})$, then $predecessor(\mathbf{b}) \subset \Omega(IPP_{\Sigma M})$.
\end{remark}

\begin{Thm}\label{thm:costsuccessor}
If a partition $\mathbf{b}$ is a successor of the partition $\mathbf{a}$ (respectively $\mathbf{b}$ is a predecessor of $\mathbf{a}$), so that $\mathbf{b}\leftarrow\mathbf{a}\oplus(i,j)$, therefore $j>i$ (resp. $i<j$) and
\begin{equation}
f(\mathbf{b})=f(\mathbf{a})+j-i>f(\mathbf{a})\text{ (resp. }<f(\mathbf{a}))
\end{equation}
\end{Thm}

\begin{proof}
\begin{eqnarray*}
f(\mathbf{b})&=&\sum \limits_{k=1}^n k b_k=\sum \limits_{k\neq i;\,k\neq j} k b_k+ib_i+jb_j\\
             &=&\sum \limits_{k\neq i;\,k\neq j} k a_k+i(a_i-1)+j(a_j+1)=\sum\limits_{k=1}^n k a_k -i +j\\
             &=&f(\mathbf{a})-i +j
\end{eqnarray*}
\end{proof}

If we list the partitions of an integer, we can compare their costs (of sum coloring) using this theorem. 
We get a comparison relation between partitions which is similar to the comparison relation between motifs used in~\cite{Lecat2017_2} (definition 4).

As an example, the Figure~\ref{fig:graph9} presents a graph with $n=9$ vertices for which it exists an unique maximum IS of size 6 ($\overline{\alpha}=\alpha(G)=6$ and $m=1$).
Moreover we take $\underline{s}$ equals to the chromatic number $\chi(G)=3$. 
Figure~\ref{fig:dag} details all the integer partitions of $n=9$ in the form of Young diagram as well as their order. 
The optimal solution of $IPP_{\Sigma M}$ is the integer partition without valid predecessor.
 
\begin{figure}[!ht]
\centering
\includegraphics[width=0.6\linewidth]{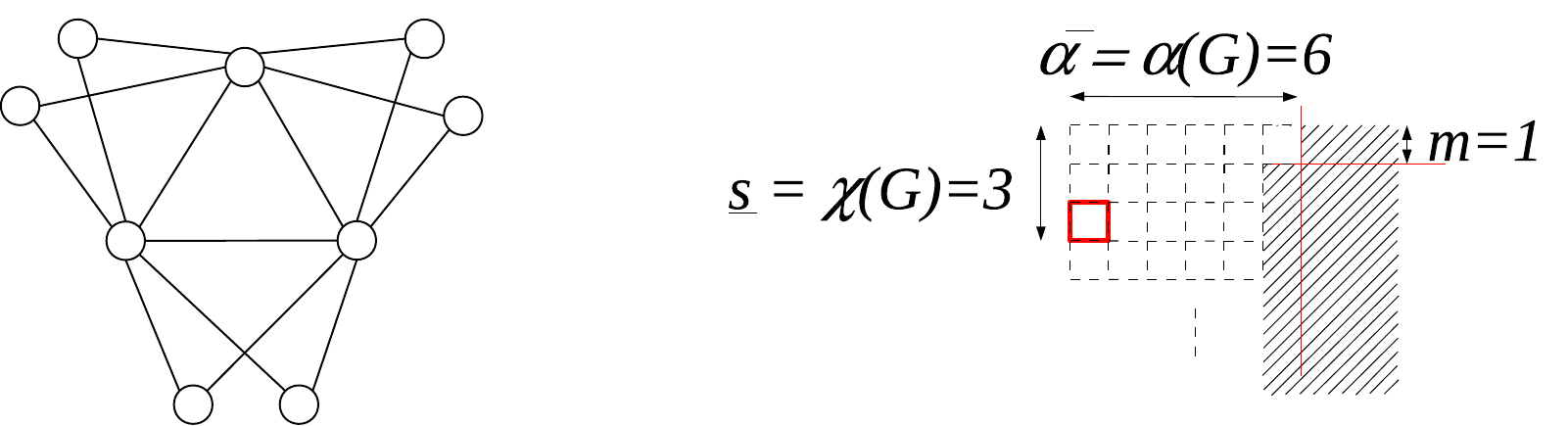}
\caption{Graph (left figure) with $n=9$ vertices with $\overline{\alpha}=\alpha(G)=6$ and $m=1$ (it exists an unique maximum IS of size 6) and we take $\underline{s}=\chi(G)=3$.
Then, the constraints on the Young diagram (right figure) imply that only one line can have 6 squares, the others has at most 5 square (i.e. no squares in the hatched area) 
and it must have at least one square on the third line (i.e. red square is mandatory).}
\label{fig:graph9}
\end{figure}

\begin{figure}[!ht]
\centering
\includegraphics[height=0.85\textheight]{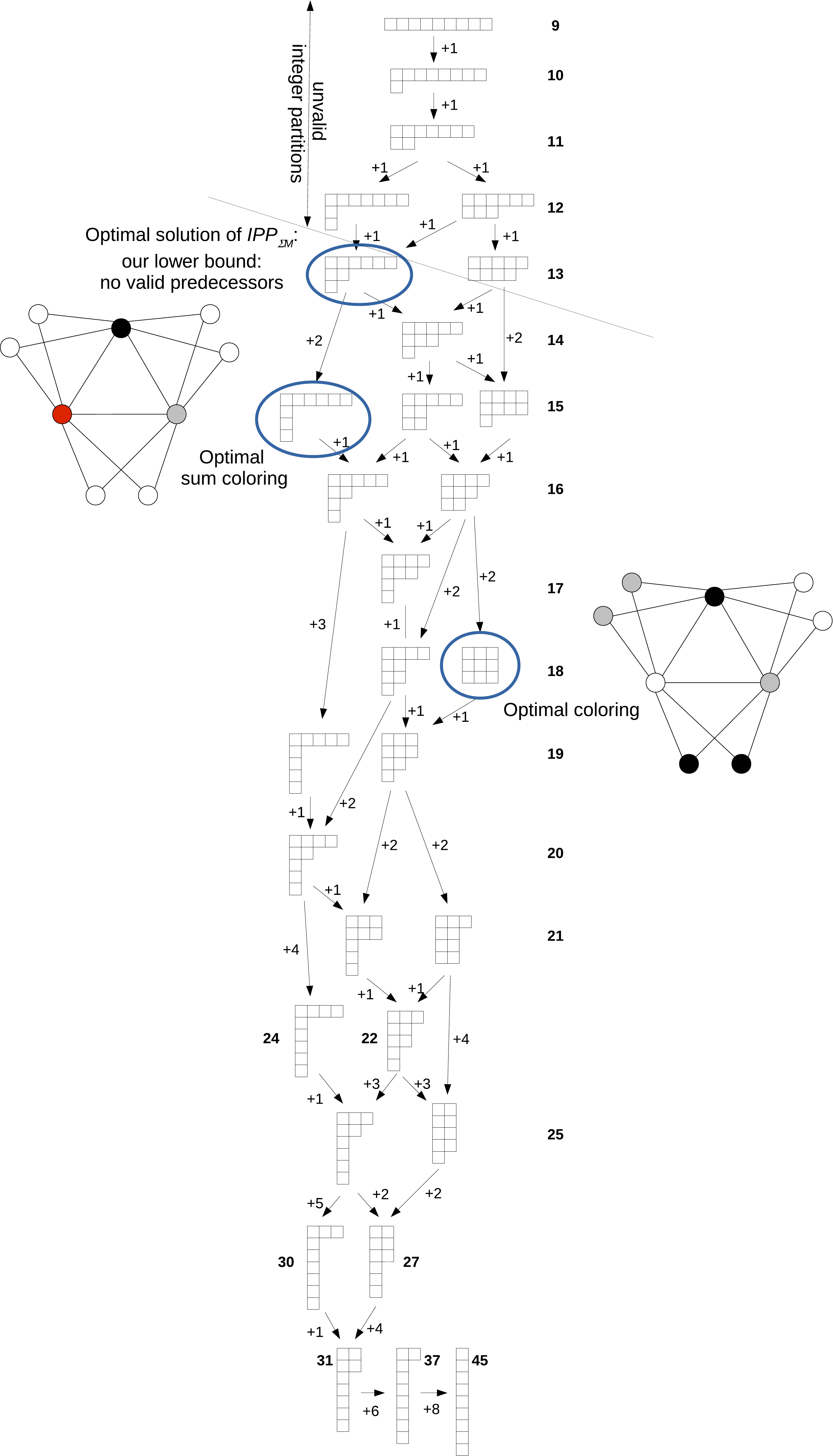}
\caption{{\bf Young Diagrams of all possible integer partitions succeeding a graph $G$.}
Each arrow indicates a successor of a partition and the number indicates the additional cost function (of sum coloring) to pass from a partition to the other.}
\label{fig:dag}
\end{figure}

\begin{Lemme}
For all $\mathbf{a}$ and $\mathbf{b}\in \Omega(IPP_{\Sigma M})$, it exists a set of $k>1$ integer partitions $\mathbf{c_i}\in \Omega(IPP_{\Sigma M})$, with $i=1...k$
so that $\mathbf{c_1}=\mathbf{a}$, $\mathbf{c_k}=\mathbf{b}$ and $\mathbf{c_{i+1}}\in predecessor(\mathbf{c_i})$ or $\mathbf{c_{i+1}}\in successor(\mathbf{c_i})$.
%From The optimal integer partition of $IPP_{\Sigma M}$ is the integer partition without predecessor.
\end{Lemme}

\begin{proof}
It means that the graph of integer partitions (the vertices are the integer partitions and an edge links two integer partitions $\mathbf{a}$ and $\mathbf{b}$ if and only if  $\mathbf{a}\in predecessor(\mathbf{b})$ or $\mathbf{a}\in successor(\mathbf{b})$) is connexe. It is evident because it exists always a path between an integer partition and the column integer partition: $(\underbrace{1,..., 1}_n)$ where $n$ is the integer to partition.
\end{proof}

\begin{Thm}
The optimal integer partition of $IPP_{\Sigma M}$ is the integer partition without predecessor.
\end{Thm}

\begin{proof}
If $\mathbf{a}$ is an optimal integer partition of $IPP_{\Sigma M}$ and has at least one predecessor $\mathbf{b}\in predessor(\mathbf{a})$, therefore $f_{\Sigma M}(\mathbf{b})<f_{\Sigma M}(\mathbf{a})$ by theorem~\ref{thm:costsuccessor} with $\mathbf{b}\in \Omega(IPP_{\Sigma M})$. It refutes the optimality of $\mathbf{a}$. 
\end{proof}

%\newpage
\subsection{Resolution of the relaxed problem $IPP_{\Sigma M}$}\label{sect:resolution}

The optimal integer partition of $IPP_{\Sigma M}$ is the integer partition without predecessor.
\subsubsection{Without constraint $a_{\underline{s}}\geq1$}

We define $(IPP_0)$ an intermediate problem corresponding to $IPP_{\Sigma M}$ but without the constraint (13): $a_{\underline{s}}\geq1$. 
Let $\mathbf{a}^*_0$ the optimal partition of $(IPP_0)$; this partition is a function of $n$, $\overline{\alpha}$ and $m$.% that returns the value of the optimal solution of $(P_1)$.
\medbreak
By definition, $\mathbf{a}^*_0(n, \overline{\alpha}, m)$ satisfies the constraints of $(IPP_0)$ and has no valid predecessor among these constraints. %Indeed the optimal partition is the partition 
To have the best cost (which implies no predecessor), the partition contains as many lines (i.e. color classes) of maximal size (equal to $\overline{\alpha}$) as possible.
$m$ corresponds to the maximum number of lines of size $\overline{\alpha}$, then we take~:
\begin{equation*}
m = \min\left(\left\lfloor\frac{n}{\alpha(G)}\right\rfloor,\ \#IS(\alpha(G)),\ \alpha(\tilde{G})\right)
\end{equation*}
where $\#IS(k)$ is the number of ISs of $G$ with size equals to $k$ a positive integer.
$\alpha(G)$, $\#IS(\alpha(G))$ and $\alpha(\tilde{G})$ can be calculated with the open source code 
%Cliquer\footnote{code available on: https://users.aalto.fi/$\sim$pat/cliquer.html}~\cite{Ostergaard2001} or 
MoMC\footnote{code available on: https://home.mis.u-picardie.fr/$\sim$cli/EnglishPage.html}~\cite{Li2017}. If it is too time-consuming, we take: $m = \left\lfloor\frac{n}{\overline{\alpha}}\right\rfloor$.

The remaining integer $n-m\times\overline{\alpha}$ uses then as many independent sets of size $\overline{\alpha}-1$ as possible.

\begin{Thm}\label{thm:cost_ws}
Let be the euclidean division of $n-m\overline{\alpha}$ by $(\overline{\alpha} -1)$~:
\begin{equation}
n-m\overline{\alpha}=q\times(\overline{\alpha}-1)+r
\end{equation}
with $q$ and $r<\overline{\alpha}-1$ two positive integers, therefore the optimal partition of $(IPP_0)$ is~:
\begin{equation}
\mathbf{a}^*_0=(\overbrace{\overline{\alpha},\ \overline{\alpha},\ ...,\ \overline{\alpha}}^m, \overbrace{\overline{\alpha} -1,...,\ \overline{\alpha} -1}^q,r)
\end{equation}
%\begin{equation}
%p^*_1(n, \overline{\alpha}, m) = m\overline{\alpha} + q(\overline{\alpha}-1) +r
%\end{equation}
and the optimal objective function is equal to~:
\begin{eqnarray}
\Sigma M_0(n, \overline{\alpha}, m)&=&f_{\Sigma M}(\mathbf{a}^*_0)\nonumber\\&=& \frac{m(m+1)}{2}\overline{\alpha} + \frac{q(2m+q+1)}{2}(\overline{\alpha}-1) + (m +q+1)r
\end{eqnarray}
\end{Thm}
\newpage
An young diagram of the optimal partition $\mathbf{a}^*_0$ is given is Figure~\ref{fig:optPartThm}.
\begin{figure}[H]
\centering
\includegraphics[height=0.3\textwidth]{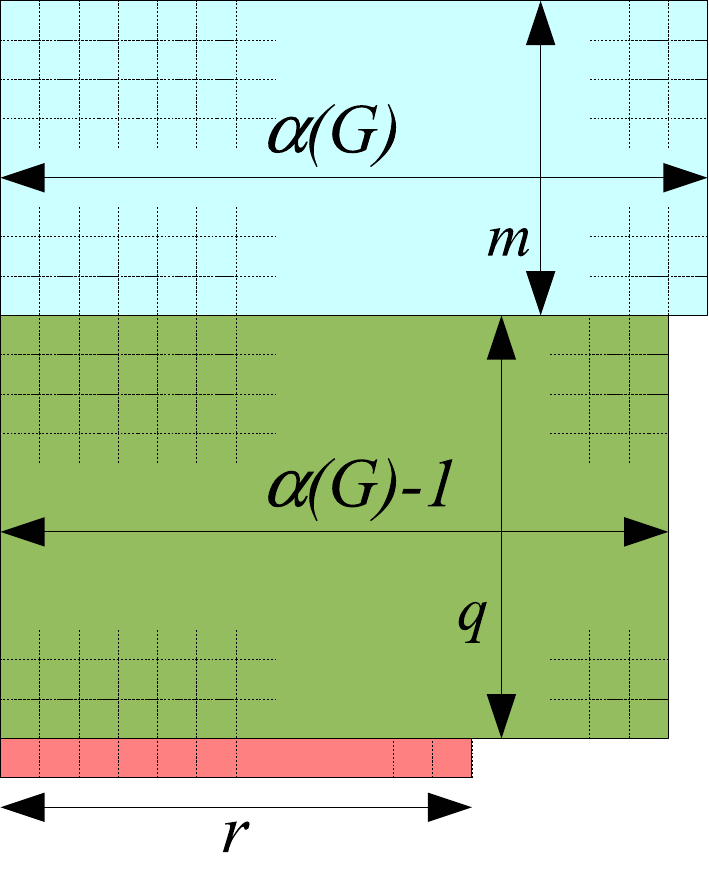}
\caption{Young diagram of the optimal partition $\mathbf{a}^*_0$ of $(IPP_0)$.}
\label{fig:optPartThm}
\end{figure}
Figure~\ref{fig:optimalPartition} gives an example of optimal partition with $n=16$, $\overline{\alpha}=\alpha(G)=4$ and $m=2$. Notice that in the worst case, when $m$ has its maximal value equals to $\left\lfloor\frac{n}{\alpha(G)}\right\rfloor$, then $q=0$ so 
$\mathbf{a}^*_0=(\overbrace{\lfloor n/\alpha(G)\rfloor...,\ \lfloor n/\alpha(G)\rfloor}^{\alpha(G)}, r)$.
%$p^*_1=\left\lfloor\frac{n}{\alpha(G)}\right\rfloor\alpha(G)+r$.
\begin{figure}[H]
\centering
\includegraphics[height=0.5\textwidth]{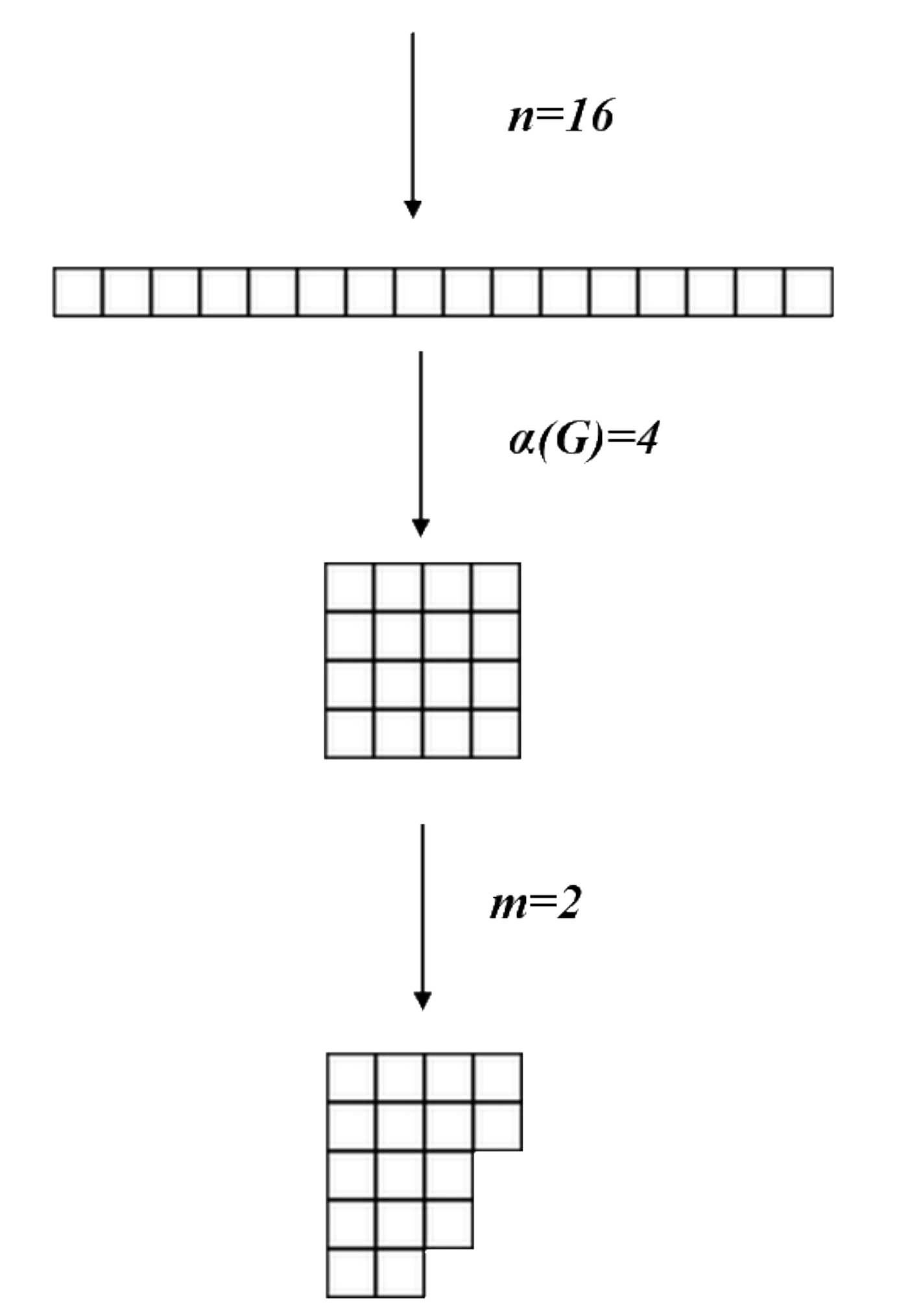}
\caption{Illustration of how to find the best partition given $n$, $\overline{\alpha}=\alpha(G)$ and $m$. The length of a line can not be over $\alpha(G)$. The number of lines with maximum lengthy can not be over $m$.}
\label{fig:optimalPartition}
\end{figure}

\subsubsection{With constraint $a_{\underline{s}}\geq1$}

If we know that we have to use at least $\underline{s}$ lines (i.e. colors), we have two possibilities :
\begin{itemize}
\item the previous solution already uses $\underline{s}$ lines, i.e. $m+q+1\geq\underline{s}$. Then the optimal solution of $(IPP_0)$ is also the optimal solution of $IPP_{\Sigma M}$, so~:
\begin{equation*}
\Sigma M = \Sigma M_0(n, \overline{\alpha}, m)=f_{\Sigma M}(\mathbf{a}^*_0)
\end{equation*}
\item the previous solution uses less than $\underline{s}$ lines. Then we start with one vertex in $\underline{s}$ different sets and we solve $(IPP_0)$ with the remaining vertices ($n-\underline{s}$) to find the solution. An illustration is given in Figure~\ref{fig:lb_decomposition}. Theorem~\ref{thm:cost_decomposition} gives the cost of the optimal partition.
\end{itemize}

\begin{figure}[H]
\centering
\includegraphics[width=0.5\textwidth]{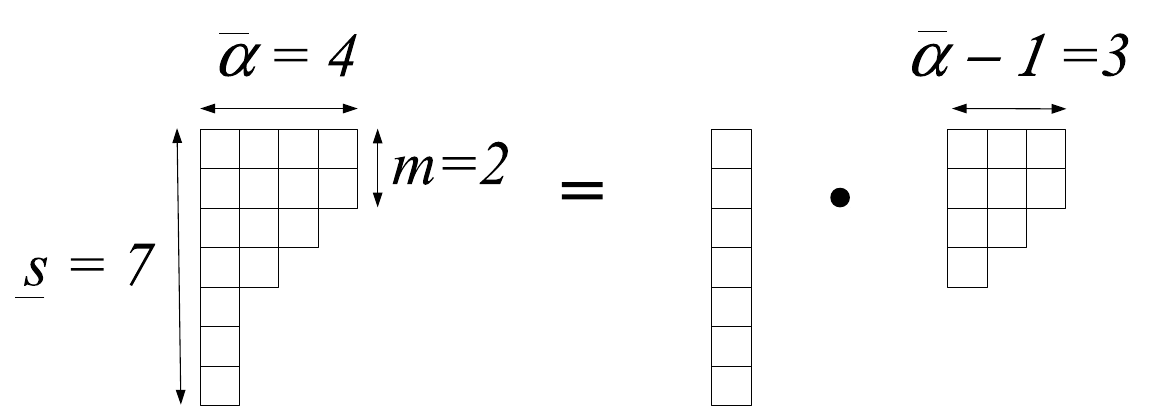}
\caption{Illustration of how to find the best partition with $n=16$, $\overline{\alpha}=4$, $m=2$ and $\underline{s}=7$. 
We add two partitions: the first one is a column of $\underline{s}$ squares; the second one is the optimal partition of one $(IPP_0)$ problem with $n'=n-\underline{s}$ squares, $\overline{\alpha}'=\overline{\alpha}-1$ and $m'=m$.}
\label{fig:lb_decomposition}
\end{figure}

\begin{Thm}\label{thm:cost_decomposition}
Let $\bullet$ be the line by line addition of two partitions. 
Let four integer $n, \overline{\alpha}, \underline{s}, m$ defining $IPP_{\Sigma M}$.
The optimal partition of $IPP_{\Sigma M}$ noted $\mathbf{a}^*$ is equal to:
\begin{equation}
\mathbf{a}^*(n, \overline{\alpha}, \underline{s}, m) = \mathbf{a}_{1\times\underline{s}} \bullet \mathbf{a}^*_0(n-\underline{s}, \overline{\alpha}-1, m)
\end{equation}
with $\mathbf{a}_{1\times\underline{s}}=(\overbrace{1,\ 1,...,\ 1}^{\underline{s}})$ the partition of the number $\underline{s}$ into ones and with $\mathbf{a}^*_0(n-\underline{s}, \overline{\alpha}-1, m)$ the optimal partition of $(IPP_0)$ with $n-\underline{s}, \overline{\alpha}-1$ and $m$ parameters, 
%\begin{equation}
%p_{\Sigma M}(n, \alpha(G), m) = (1 \times \underline{s}(G)) \bullet p_1 (n - \underline{s}(G), \alpha(G) -1, m)
%\end{equation}
therefore: 
\begin{equation}
\Sigma M = \frac{\underline{s} (\underline{s}+1)}{2} + \Sigma M_0(n - \underline{s}, \overline{\alpha}-1, m)
\end{equation}
\end{Thm}

\section{Lower bound for graph coloring}\label{sect:LB_GVCP}
The relaxation used for MSCP can also be applied for GVCP. 
GVCP is a special case of MSCP for which all colors have the same cost, then let be a possible formulation of GVCP~:
\begin{equation*}
(GVCP)\left\{
\begin{minipage}[c]{0.8\linewidth}\vspace{-5mm}
\begin{eqnarray}%{rll}
\displaystyle \text{Min.}   &   \displaystyle f_\chi(\mathbf{V})= |\{V_l\ |\ |V_l|>0,\ \forall l=1..n\}| \label{eq:obj2}\\[-0.5mm]
\text{s.c.} &Eq.~(\ref{eq:partition1}-\ref{eq:consColoring})\nonumber\\[-0.5mm]%\label{eq:partition1}
%            &\displaystyle V_l \cap V_k = \emptyset,&\forall (l,k), l\neq k\\[-0.5mm]%
%            &\displaystyle |V_l| \geq |V_{l+1}|, &\forall l=1..n-1\\[-0.5mm]%\label{eq:decrease}
%            &\displaystyle \{i,j\} \nsubseteq V_l,&\forall (i,j) \in E,\ \forall l=1..n\\[-0.5mm]%\label{eq:consColoring}
            & V_l \subset V&\forall l=1..n\nonumber%
%\text{et}           & x \in \N^n \text{ ou } x \in \{0,1\}^n \\[1mm]
\end{eqnarray}
\end{minipage}
\right. \nonumber
\end{equation*}
This formulation differs from MSCP just by the objective function, that counts the number of non-empty partitions.
Therefore, the relaxation into an integer partition problem becomes~:
\begin{equation*}
(IPP_{\chi})\left\{
\begin{minipage}[c]{0.7\linewidth}\vspace{-5mm}
\begin{eqnarray}
\displaystyle \text{Min.}   &   \displaystyle f_{\chi}(\mathbf{a})= |\{a_i | a_i>0,\ \forall i=1..n\}| \\[-0.5mm]
\text{s.c.} &Eq.~(\ref{eq:partition1bis}-\ref{eq:constraintm}) \nonumber\\[-0.5mm]
%            &\displaystyle a_i \geq a_{i+1}, &\forall i=1..n-1\\[-0.5mm]
%            &\displaystyle a_i \leq \overline{\alpha},&\forall i=1..n\label{eq:coloringyy}\\[-0.5mm]
%            &\displaystyle |\{a_i | a_i=\overline{\alpha}\}| \leq m\\[-0.5mm]
%            &a_{\underline{s}}\geq 1\\[-0.5mm]
            & \mathbf{a}=(a_i)_{1\leq i\leq n} \in \mathbb{N}^n \nonumber
%\text{et}           & x \in \N^n \text{ ou } x \in \{0,1\}^n \\[1mm]
\end{eqnarray}
\end{minipage}
\right. \nonumber
\end{equation*}
Again $IPP_{\chi}$ differs from $IPP_{\Sigma M}$ just by the objective function, that corresponds to the number of lines in Young diagram and because Eq.~(\ref{eq:LBchi}) is useless.
In this case, the optimal objective function value is equal to~: 
\begin{equation}
LB_\chi=f_{\chi}(\mathbf{a^*})=m+q+[r\neq0] 
\end{equation}
with $m$, $q$ and $r$ given by the Thm.\ref{thm:cost_ws} and $\mathbf{a^*}$ the optimal solution of~$IPP_{\chi}$ 
and $[\ ]$ is Iverson  bracket~: $[r\neq0] = 1$ if $r\neq0$ is true and equals $0$ otherwise. 
%function that is equal to 1 if $<condition>$ is true and is equal to zero otherwise.
$LB_\chi$ is a lower bound of $\chi(G)$ and its value depends only of three integers: $n$, $\overline{\alpha}$ and $m$.

\section{Results}\label{sect:result}

\subsection{Procedure}
To compute our lower bound $\Sigma M$, we use $\overline{G}$  the complementary graph of $G$. We run MoMC solver
%\"{O}sterga\aa rd algorithm~\cite{Ostergaard2001}, 
to find $\alpha(G)$, the size of the maximal clique of $\overline{G}$. We also save the list of all the maximal IS of $G$ to compute $\#is$ and to build $\tilde{G}$.
%Then we follow the same procedure with the graph $\tilde{G}$ of the independent sets to find $m=\alpha (\tilde{G})$.
Algorithm~\ref{algo:global_procedure} recall the global procedure to compute $\Sigma M$.

\begin{algorithm}[H]
 \DontPrintSemicolon
% \SetVline
\Input{$G$, $n$}
% \\~\\
  %Compute $\overline{\alpha}(G)$\textit{ // $=\alpha(G)$ if it is possible with cliquer}\;
  $\overline{\alpha}\leftarrow
\left\{
\begin{array}{ll}
\alpha(G)& \text{if it is not too time-consuming with MoMC solver,}\\
\overline{\alpha}(G)&\text{otherwise.}
\end{array}
\right.$\;
  %Compute $\#is(\overline{\alpha}(G))$\textit{ // with cliquer if it is possible}\;
  $\#is\leftarrow
\left\{
\begin{array}{ll}
\#is(\overline{\alpha})& \text{if it is not too time-consuming with MoMC solver,}\\
\left\lfloor\frac{n}{\overline{\alpha}}\right\rfloor&\text{otherwise.}
\end{array}
\right.$\;
  Build $\tilde{G}$, the graph of the maximum IS.\;
  %Compute $\alpha(\tilde{G})\textit{ // with cliquer if it is possible}$\;
  $\tilde{\alpha}\leftarrow
\left\{
\begin{array}{ll}
\alpha(\tilde{G})& \text{if it is not too time-consuming with MoMC solver,}\\
\#is&\text{otherwise.}
\end{array}
\right.$\;
  %$m \leftarrow \min\left(\left\lfloor\frac{n}{\alpha(G)}\right\rfloor,\ \#is(\alpha(G)),\ \alpha(\tilde{G})\right)$\;
  $m \leftarrow \min\left(\left\lfloor\frac{n}{\overline{\alpha}}\right\rfloor,\ \#is,\ \tilde{\alpha}\right)$\;
  %$m \leftarrow \tilde{\alpha}$\;
  %Compute $q,r$ so that $n-m\alpha(G)=q\times(\alpha(G)-1)+r$\;
  %Compute $q,r$ so that $n-m\overline{\alpha}=q\times(\overline{\alpha}-1)+r$\;
  $q\leftarrow \left\lfloor\frac{n-m\overline{\alpha}}{\overline{\alpha}-1}\right\rfloor$\;
  $r\leftarrow n-m\overline{\alpha} - q\times(\overline{\alpha}-1)$\;
  $\underline{s}\leftarrow
\left\{
\begin{array}{ll}
\underline{\chi}(G)&\text{the best known lower bound of $\chi(G)$ if it is known,}\\
\left\lceil\frac{n}{\overline{\alpha}}\right\rceil&\text{otherwise.}
\end{array}
\right.$\;
  $\Sigma M\leftarrow
\left\{
\begin{array}{ll}
\Sigma M_0(n,\overline{\alpha},m)& \text{if } m+q+1\geq\underline{s}\\
\frac{\underline{s}(\underline{s}+1)}{2}+\Sigma M_0(n-\underline{s},\overline{\alpha}-1,m)&\text{otherwise}
\end{array}
\right.$\;
\caption{Procedure to compute $\Sigma M$.}\label{algo:global_procedure}
\end{algorithm}

\subsection{Empirical Results and Analysis}
We tested our procedure on some graph instances of DIMACS and COLOR  benchmarks, which are frequently used for performance evaluation of $MSCP$ algorithms~\cite{Jin2016}.
For several instances, the list of maximum independent set takes too much computational time to be found. There are two possible reasons for this : 
\begin{itemize}
 \item $\overline{G}$ is too dense and MoMC needs too much time to find a maximum clique.
 \item the number of maximum cliques is too high and we can't compute their graph.
\end{itemize}

In our tests, we reuse open source code (MoMC and Cliquer) written in C and we write a script coded in Python~\footnote{code available on: github.com/gondran/LowBoundSumColoring}.
The results were obtained with an Intel Xeon E5 2.50GHz processor - 8 cores and 16GB of RAM.

In the Table~\ref{tab:results}, we show a set of results, focused on graphs where our lower bound could be computed. We compare our lower bound, $\Sigma M$, to $\Sigma_{old}$ (the best lower bound known according to \cite{Jin2017}, \cite{Lecat2017_2}, \cite{Wu2018}) and $LBM\Sigma$ (lower bound described  in \cite{Lecat2017_2}, that we calculate or update for several graphs).
The three first columns indicate the graph instance name, its number of vertices and its density.
Columns 4 and 5 give the maximum independent set of the graph, $\alpha(G)$ and the time (in second) required to compute it with the MoMC solver.
The two next columns show the number of maximum independent set of the graph, $\string#is$ and the time (in second) required to compute it with the MoMC solver or the Cliquer solver.
Columns 8 and 9 indicate, $m$, the maximum number of \emph{compatibles} independent sets of the graph and the time (in second) required to compute it with the MoMC solver. If $\string#is$ is too high ($>5 000$), then the computation of $m$ is not done because it would take too long.
Column 10 gives the chromatic number of the graph $\chi(G)$, if it is known and its interval of belonging otherwise.
Note that this information comes from other algorithms after a computation-time that may be very long. 
Next column provides the result of our lower bound, $LB_\chi$, computed as in section~\ref{sect:LB_GVCP}.
Last five columns display, $\Sigma(G)$, the chromatic sum when it is known; ${\Sigma}_{old}$, the best lower bound of $\Sigma(G)$ known in the literature; 
$LBM\Sigma$, the lower bound presented in~\cite{Lecat2017_2};
$\Sigma M_0$, our lower bound computed without $\underline{s}=\underline{\chi}(G)$; 
$\Sigma M$, our lower bound computed with $\underline{s}=\underline{\chi}(G)$.

% Place tables after the first paragraph in which they are cited.
%\begin{table}[!ht]
% \begin{adjustwidth}{-2.25in}{0in} % Comment out/remove adjustwidth environment if table fits in text column.
% \centering
% \caption{
% {\bf Table caption Nulla mi mi, venenatis sed ipsum varius, volutpat euismod diam.}}
% \begin{tabular}{|l+l|l|l|l|l|l|l|}
% \hline
% \multicolumn{4}{|l|}{\bf Heading1} & \multicolumn{4}{|l|}{\bf Heading2}\\ \thickhline
% $cell1 row1$ & cell2 row 1 & cell3 row 1 & cell4 row 1 & cell5 row 1 & cell6 row 1 & cell7 row 1 & cell8 row 1\\ \hline
% $cell1 row2$ & cell2 row 2 & cell3 row 2 & cell4 row 2 & cell5 row 2 & cell6 row 2 & cell7 row 2 & cell8 row 2\\ \hline
% $cell1 row3$ & cell2 row 3 & cell3 row 3 & cell4 row 3 & cell5 row 3 & cell6 row 3 & cell7 row 3 & cell8 row 3\\ \hline
% \end{tabular}
% \begin{flushleft} Table notes Phasellus venenatis, tortor nec vestibulum mattis, massa tortor interdum felis, nec pellentesque metus tortor nec nisl. Ut ornare mauris tellus, vel dapibus arcu suscipit sed.
% \end{flushleft}
% \label{table1}
% \end{adjustwidth}
% \end{table}

\newpage
\begin{table*}[h!]
%\nothal{\begin{adjustwidth}{-1in}{0in}} % Comment out/remove adjustwidth environment if table fits in text column.
\hal{\begin{adjustwidth}{-0.3in}{0in}}{\begin{adjustwidth}{-1in}{0in}} % Comment out/remove adjustwidth environment if table fits in text column.
\centering
\setlength{\tabcolsep}{3pt}
%	\begin{longtable}{ccccccccccc}
\caption{\label{tab:results}Results of \emph{weak optimality} tests on some graphs of DIMACS and COLOR benchmarks.}
{\footnotesize
    \begin{tabular}{cccccccccccccccc}
\toprule
Instances & $n$ & $d$& $\alpha(G)$&time&$\string#is$&time&$m$&time&$\chi(G)$&$LB_\chi$& $\Sigma(G)$ & ${\Sigma}_{old}$& $LBM\Sigma$&$\Sigma M_0$&$\Sigma M$\\
\midrule
\bottomrule
myciel3 & 11 & 0.36 & 5 & $<0.1$  & 2 & $<0.1$ & 1 & 0 & 4 & 3 & ? &\bf 20 & \bf20 & 19 & \bf20\\ 
myciel4 & 23 & 0.28 & 11 & $<0.1$ & 2 & $<0.1$ & 1 & 0 & 5 & 3 & ? & \bf41 & \bf41 & 37 & \bf41\\ 
myciel5 & 47 & 0.22 & 23 & $<0.1$ & 1 & $<0.1$ & 1 & 0 & 6 & 3 & ? & 80 & \bf81 & 73 & \bf81\\ 
myciel6 & 95 & 0.17 & 47 & $<0.1$ & 1 & $<0.1$ & 1 & 0 & 7 & 3 & ? & \bf158 & \bf158 & 145 & \bf158\\ 
myciel7 & 191 & 0.13 & 95 & $<0.1$ & 1 & 0.1 & 1 & 0 & 8 & 3 & ? & \bf308 & \bf308 & 289 & \bf308\\ 

queen5\_5 & 25 & 0.53 & 5 & $<0.1$ & 10 & $<0.1$ & 5 & $<0.1$ & 5 & \bf5 & \bf75 & \bf75 & \bf75 & \bf75 & \bf75\\ 
\rowcolor{lightblue}queen6\_6 & 36 & 0.46 & 6 & $<0.1$ & 4 & $<0.1$ & 4 & $<0.1$ & 7 & \bf7 & ? & 126 & 127 & \bf129 & \bf129\\ 
queen7\_7 & 49 & 0.40 & 7 & $<0.1$ & 40 & $<0.1$ & 7 & $<0.1$ & 7 & \bf7 & \bf196 & \bf196 & \bf196 & \bf196 & \bf196\\ 
\rowcolor{lightblue}queen8\_8 & 64 & 0.36 & 8 & $<0.1$ & 92 & $<0.1$ & 6 & $<0.1$ & 9 & \bf9 & \bf291* & 288 & 289 & \bf291 & \bf291\\ 
queen8\_12 & 96 & 0.30 & 8 & $<0.1$ & 195 271 & 0.3 & $\string#is$& $\times$& 12 & \bf12 & ? & \bf624 & \bf624 & \bf624 & \bf624\\ 
\rowcolor{lightblue}queen9\_9 & 81 & 0.33 & 9 & $<0.1$ & 352 & $<0.1$ & 7 & $<0.1$ & 10 & \bf10 & ? & 405 & 406 & \bf408 & \bf408\\ 
\rowcolor{lightblue}queen10\_10 & 100 & 0.30 & 10 & $<0.1$ & 724 & 0.2 & 8 & 0.8 & 11 & \bf11 & \bf553* & 550 & 551 & \bf553 & \bf553\\ 
queen11\_11 & 121 & 0.27 & 11 & $<0.1$ & 2 680 & 0.9 & 11 & 207 & 11 & \bf11 & \bf726 & \bf726 & \bf726 & \bf726 & \bf726\\ 
queen12.12&144&0.25&12&$<0.1$&14 200&6**&$\string#is$&$\times$&12&\bf12&\bf 936&\bf 936&\bf 936&\bf 936&\bf 936\\
queen13.13&169&0.23&13&$<0.1$&73 712&0.5**&$\string#is$&$\times$&13&\bf13&\bf 1 183&\bf 1 183&\bf 1 183&\bf 1 183&\bf 1 183\\
queen14.14&196&0.22&14&$<0.1$&365 596&3**&$\string#is$&$\times$&14&\bf14&\bf 1 470&\bf 1 470&\bf 1 470&\bf 1 470&\bf 1 470\\
queen15.15&225&0.21&15&$<0.1$&2 279 184&17**&$\string#is$&$\times$&15&\bf15&\bf 1 800&\bf 1 800&\bf 1 800&\bf 1 800&\bf 1 800\\
queen16.16&256&0.19&16&$<0.1$&14 772 512&113**&$\string#is$&$\times$&16&\bf16&?&\bf 2 176&\bf 2 176&\bf 2 176&\bf 2 176\\

2-Insertions\_3 & 37 & 0.11 & 18 & $<0.1$  & 1 & $<0.1$ & 1 & 0 & 4 & 3 & ? & 55 &\bf 59 & 58 &\bf 59\\ 
\rowcolor{lightblue}3-Insertions\_3 & 56 & 0.07 & 27 & $<0.1$ & 11 & $<0.1$ & 1 & 0 & 4 & 3 & ? & 84 & 88 & 88 & \bf89\\ 

\rowcolor{lightblue}DSJC125.1 & 125 & 0.09 & 34 & $<0.1$ & 747 & $<0.1$ & 1 & 0 & 5 & 4 & ? & 297 & 297 & 299 & \bf300\\ 
\rowcolor{lightblue}DSJC125.5 & 125 & 0.50 & 10 & $<0.1$ & 2 & $<0.1$ & 1 & 0 & 17 & 14 & ? & 851 & 855 & 918 & \bf924\\ 
\rowcolor{lightblue}DSJC125.9 & 125 & 0.90 & 4 & $<0.1$ & 9 & $<0.1$ & 5 & $<0.1$ & 44 & 40 & ? & 2 108 & 2 124 & 2 475 & \bf2 487\\ 
\rowcolor{lightblue}DSJC250.5 & 250 & 0.50 & 12 & $<0.1$ & 2 & $<0.1$ & 2 & $<0.1$ & \textlbrackdbl$26,\ 28$\textrbrackdbl & 23 & ? & 2 745 & 2 745 & 2 924 & \bf2 930\\ 
\rowcolor{lightblue}DSJC250.9 & 250 & 0.90 & 5 & $<0.1$ & 3 & $<0.1$ & 2 & $<0.1$ & 72 & 62 & ? & 6 651 & 6 678 & 7 815 & \bf7 882\\
\rowcolor{lightblue}DSJC500.5 & 500 & 0.50 & 13 & 2 & 51 & 3 & 9 & $<0.1$ & \textlbrackdbl$43,\ 47$\textrbrackdbl & 41 & ? & 9 867 & 9 877 & 10 336 & \bf10 339\\ 
\rowcolor{lightblue}DSJC500.9 & 500 & 0.90 & 5 & $<0.1$ & 23 & $<0.1$ & 15 & $<0.1$ & \textlbrackdbl$123,\ 126$\textrbrackdbl & 122 & ? & 25 581 & 25 581 & 29 766 & \bf29 768\\ 
\rowcolor{lightblue}DSJC1000.5&1000&0.50&15&159&12&290&6&$<0.1$&\textlbrackdbl$73,\ 82$\textrbrackdbl&71&?&33 835& 33 856&35 805&{\bf 35 808}\\
\rowcolor{lightblue}DSJC1000.9 & 1000 & 0.90 & 6 & 0.1 & 3 & 0.1 & 3 & $<0.1$ & \textlbrackdbl$216,\ 222$\textrbrackdbl & 200 & ? & 85 235 & 85 294 & 99 906 & \bf100 078\\ 

DSJR500.1c & 500 & 0.97 & 13 & $<0.1$ & 4 & $<0.1$ & 2 & $<0.1$ & 85 & 42 & ? & \bf15 398 & 11 040 & 10 587 & 11 619\\ 
DSJR500.5 & 500 & 0.47 & 7 & 0.3 & 18 & 0.8 & 2 & $<0.1$ & 122 & 83 & ? & \bf23 609 & 19 599 & 20 919 & 21 832\\ 

flat300\_20\_0&300&0.48&15&$<0.1$&20&$<0.1$&20&$<0.1$&20&\bf 20&\bf 3 150&\bf 3 150&\bf 3 150&\bf 3 150&\bf 3 150\\
\rowcolor{lightblue}flat300\_26\_0&300&0.48&12&$<0.1$&31&1&14&$<0.1$&26&\bf 26&\bf 3 966*&3 901&3 901&\bf3 966&\bf3 966\\
\rowcolor{lightblue}flat300\_28\_0&300&0.48&12&$<0.1$&45&1&6&$<0.1$&28&27&?&3 906&3 906&4 098&\bf 4 099\\
flat1000\_50\_0&1000&0.49&20&36&50&78&50&$<0.1$&50&\bf50&\bf 25 500&\bf 25 500&\bf 25 500&\bf 25 500&\bf 25 500\\
\rowcolor{lightblue}flat1000\_60\_0&1000&0.49&17&89&42&199&40&$<0.1$&60&\bf60&\bf 30 100*&29 914&29 914&{\bf 30 100}&{\bf 30 100}\\
\rowcolor{lightblue}flat1000\_76\_0&1000&0.49&15&184&21&394&8&$<0.1$&76&71&?&33 880&33 880&35678&{\bf 35 693}\\
	\end{tabular}}
%\nothal{\end{adjustwidth}}
\end{adjustwidth}
* new optimal solution that we proved

** computation was done with the Cliquer\footnote{code available on: https://users.aalto.fi/$\sim$pat/cliquer.html}~\cite{Ostergaard2001} 
 solver because for listing all maximum cliques Cliquer is faster than MoMC. 
\end{table*}
%\newline
%** optimal solution yet known
%PLOS does not support heading levels beyond the 3rd (no 4th level headings).

On the 37 graphs of this test set, 
we improve the result for 18 graphs (light blue in Table~\ref{tab:results}). 
%the 18 blue lines indicate the 18 graphs that improve the known lower bound.
For 17 graphs we find the same lower bound as in the literature.
%We equalize the result for 17 graphs. 
For 2 graphs, our lower bound is worse than the best lower bound of the literature.
%We find a worse result for 2 graphs. 
The reason is that for these graphs, the approach used by Moukrim et al.~\cite{Moukrim2010,Moukrim2013} is totally different of ours, based on the decomposition of the graph into partition of cliques. 
For 9 graphs, the optimal lower bound was already found with the much simpler $LBM\Sigma$ lower bound.
We have proven the exact value of $\Sigma(G)$ for 4 graphs (with * in Table~\ref{tab:results}), and have computed an example of optimal coloring.
%The 18 blue lines indicate the 18 graphs that improve the known lower bound.
Our approach is similar to that of Lecat et al.~\cite{Lecat2017_2} but we outperformed their results by introducing an additional constraint with the integer $m$.
%$m$ fixes the maximum number of lines of size $\alpha(G)$. 
%For 14 of those 18 graphs, the value of $m$ is define through the definition of $\tilde{G}$ graph.
%Moreover, we prove the optimality of the $MSCP$ for 4 graphs (with * in Table~\ref{tab:results}).

Our lower bound of $\chi(G)$, $LB_\chi$, is equal to the chromatic number for all \emph{queen} graph instances (13 graphs) and for 4 \emph{flat} graph instances.
Note that for five of those instances (\emph{queen5\_5}, \emph{queen7\_7}, \emph{queen11\_11}, \emph{flat300\_20\_0} and \emph{flat1000\_50\_0}), we also found the optimal coloring.
Indeed, in those cases, the decomposition of the number of vertices is~: $n=\alpha(G)\times m+(\alpha(G)-1)\times q+r$ with $q=r=0$, i.e. $\chi(G)=m$.
Therefore we know the $m$ compatible independent sets that composed the optimal solution.

Moreover, our $LB_\chi$ is computed in $78~s$ and $199~s$ for respectively \emph{flat1000\_50\_0} and \emph{flat1000\_60\_0} graph instances while it takes respectively $3~331s$ and $29~996s$ with the most powerful method~\cite{held2012} to find a lower bound of $\chi(G)$.
%For \emph{flat1000\_76\_0} graph instance, \cite{held2012} find a lower bound value of 72 in $190~608~s$ (more than 2 days), while our $LB_chi$ find 71 in $394~s$

\section{Conclusion}\label{conclusion}

We presented a new way to find lower bounds for the $MSCP$ and the $GVCP$. 
In order to do this, we explained how to relax $MSCP$ into an integer partition problem that can be exactly solved. 
We can select the constraints we want to keep in this integer partition problem, and we proposed a set of constraints used to define the $\Sigma M$ lower bound.
We carried out experiments and improved the best known lower bound for 18 graphs of standard benchmark DIMACS. 
We also proved the optimality of 4 of them.

It is also possible to add even more constraints in the integer partition problem. Further researches could use this approach to keep improving the lower bound.

\nolinenumbers

\bibliographystyle{elsarticle-num}
%\bibliography{mybibfile}
\bibliography{bib_sumcol,gcp_bib}

\begin{thebibliography}{10}
\expandafter\ifx\csname url\endcsname\relax
  \def\url#1{\texttt{#1}}\fi
\expandafter\ifx\csname urlprefix\endcsname\relax\def\urlprefix{URL }\fi
\expandafter\ifx\csname href\endcsname\relax
  \def\href#1#2{#2} \def\path#1{#1}\fi

\bibitem{Malafiejski2004}
M.~Malafiejski, Sum coloring of graphs, in: Kubale  \cite{Kubale2004}, pp.
  55--66.

\bibitem{Kubale2004}
M.~Kubale (Ed.), Graph Colorings, Vol. 352 of DIMACS Series in Discrete
  Mathematics and Theoretical Computer Science, American Mathematical Society,
  Providence, Rhode Island, USA, 2004.

\bibitem{Wu2018}
Q.~Wu, Q.~Zhou, Y.~Jin, J.~Hao, {Minimum sum coloring for large graphs with
  extraction and backward expansion search}, Appl. Soft Comput. 62 (2018)
  1056--1065.
\newblock \href {http://dx.doi.org/https://doi.org/10.1016/j.asoc.2017.09.043}
  {\path{doi:https://doi.org/10.1016/j.asoc.2017.09.043}}.

\bibitem{Jin2016}
Y.~Jin, J.-K. Hao, {Hybrid evolutionary search for the minimum sum coloring
  problem of graphs}, Information Sciences (2016) 15--34.

\bibitem{Kubicka1991}
E.~Kubicka, G.~Kubicki, D.~Kountanis, {Approximation algorithms for the
  chromatic sum}, in: {Computing in the 90's}, Springer, 1991, pp. 15--21.

\bibitem{Jin2017}
Y.~Jin, J.-P. Hamiez, J.-K. Hao, {Algorithms for the minimum sum coloring
  problem: a review}, Artificial Intelligence Review (2017) 367--394.

\bibitem{Lecat2017_2}
C.~Lecat, C.~Lucet, C.-M. Li, {New Lower Bound for the Minimum Sum Coloring
  Problem.}, in: {AAAI}, 2017, pp. 853--859.

\bibitem{dimacs96}
D.~S. Johnson, M.~Trick (Eds.), Cliques, Coloring, and Satisfiability: Second
  {DIMACS} Implementation Challenge, 1993, Vol.~26 of DIMACS Series in Discrete
  Mathematics and Theoretical Computer Science, American Mathematical Society,
  Providence, RI, USA, 1996.

\bibitem{Wu2013}
Q.~Wu, J.-K. Hao,
  \href{http://dblp.uni-trier.de/db/journals/corr/corr1303.html#abs-1303-6761}{Improved
  lower bounds for sum coloring via clique decomposition}, CoRR abs/1303.6761.
\newline\urlprefix\url{http://dblp.uni-trier.de/db/journals/corr/corr1303.html#abs-1303-6761}

\bibitem{Young1900}
A.~Young, {On quantitative substitutional analysis}, Proceedings of the London
  Mathematical Society 33 (1900) 97--145.

\bibitem{Garey79}
M.~R. Garey, D.~S. Johnson, Computers and Intractability: A Guide to the Theory
  of {${\cal NP}$}-Completeness, Freeman, San Francisco, CA, USA, 1979.

\bibitem{Valiant1979}
L.~G. Valiant, The complexity of enumeration and reliability problems, {SIAM}
  J. Comput. 8~(3) (1979) 410--421.
\newblock \href {http://dx.doi.org/10.1137/0208032}
  {\path{doi:10.1137/0208032}}.

\bibitem{Li2017}
C.-M. Li, H.~Jiang, F.~Many\'{a}, On minimization of the number of branches in
  branch-and-bound algorithms for the maximum clique problem, Computers \&
  Operations Research 84 (2017) 1--15.
\newblock \href {http://dx.doi.org/https://doi.org/10.1016/j.cor.2017.02.017}
  {\path{doi:https://doi.org/10.1016/j.cor.2017.02.017}}.

\bibitem{Ostergaard2001}
P.~R. {\"O}sterg{\aa}rd, {A new algorithm for the maximum-weight clique
  problem}, Nordic Journal of Computing (2001) 424--436.

\bibitem{Moukrim2010}
A.~Moukrim, K.~Sghiouer, C.~Lucet, Y.~Li, {Lower Bounds for the Minimal Sum
  Coloring Problem}, Electronic Notes in Discrete Mathematics 36 (2010)
  663--670, iSCO 2010 - International Symposium on Combinatorial Optimization.

\bibitem{Moukrim2013}
A.~Moukrim, K.~Sghiouer, C.~Lucet, Y.~Li,
  \href{https://www.hds.utc.fr/~moukrim/dokuwiki/_media/en/mscp_cor13septembre2013.pdf}{{Upper
  and Lower Bounds for the Minimum Sum Coloring Problem}}, Tech. rep.,
  Universit{\'e} de Technologie de Compi{\`e}gne and Universit{\'e} de Picardie
  Jules Verne (2013).
\newline\urlprefix\url{https://www.hds.utc.fr/~moukrim/dokuwiki/_media/en/mscp_cor13septembre2013.pdf}

\bibitem{held2012}
S.~Held, W.~Cook, E.~Sewell, Maximum-weight stable sets and safe lower bounds
  for graph coloring, Mathematical Programming Computation 4~(4) (2012)
  363--381.
\newblock \href {http://dx.doi.org/10.1007/s12532-012-0042-3}
  {\path{doi:10.1007/s12532-012-0042-3}}.

\end{thebibliography}

\end{document}